\documentclass[12pt]{article}
\usepackage{amsmath,amssymb}

\usepackage{enumerate}

\usepackage{graphicx}

\usepackage{makeidx}
\usepackage[utf8]{inputenc}
\usepackage{wrapfig}
\usepackage{amsmath}
\usepackage{amsfonts}
\usepackage{mathrsfs}
\usepackage{amssymb}
\usepackage{latexsym}
\usepackage{amsbsy}
\usepackage{color}
\usepackage{bbm}

\usepackage{mathrsfs}

\usepackage{wasysym}

\providecommand{\otherindexspace}[1]{}

\usepackage{amsthm}

\newtheorem{theorem}{Theorem}[section]

\newtheorem{lemma}[theorem]{Lemma}
\newtheorem{proposition}[theorem]{Proposition}
\newtheorem{remark}[theorem]{Remark}

\newtheorem{definition}[theorem]{Definition}

\newtheorem{corollary}[theorem]{Corollary}

\newtheorem{assumption}[theorem]{Assumption}

\numberwithin{equation}{section}

\def\cal#1{\mathcal{#1}}

\def \H{\mathbb {H}}
\def \R{\mathbb {R}}

\usepackage{fancyhdr}
\newcommand{\dproof}{\noindent {Proof.} \quad}
\newcommand{\fproof}{\hfill $\square$ \bigskip}

\def \H{I\!\!H}

\def \R{I\!\!R}

\def\Fc{{\cal F}}

\def\Dzw1#1{\frac{\partial^2 #1}{\partial z \partial w_1}}

\def\Dzb1#1{\frac{\partial^2 #1}{\partial z \partial b_1}}

\def\EB{\mathbb{E}}

\def\R{{\bf R}}

\def\1B{\text{1\!\!I}}

\def\EB{\mathbb{E}}

\def\R{{\bf R}}

\def\1B{\text{1\!\!I}}

\def\EB{\mathbb{E}}

\usepackage{graphics}
\usepackage{graphicx}

\DeclareGraphicsExtensions{.eps,.bmp,.jpg,.pdf,.mps,.png,.gif}

\makeatletter
\def\titre{\@title}
\makeatother

\title{BSDEs with default jump}

\author{Roxana Dumitrescu\thanks{Department of Mathematics, King's College London, United Kingdom, email: \textbf{roxana.dumitrescu@kcl.ac.uk}} \and Marie-Claire Quenez \thanks{LPMA,
Université Paris 7 Denis Diderot, Boite courrier 7012, 75251 Paris cedex 05, France, email: \textbf{quenez@math.univ-paris-diderot.fr}} \and  Agnès Sulem
\thanks{ INRIA Paris, France, and Université Paris-Est, email: \textbf{agnes.sulem@inria.fr}}}

\begin{document}



\date{\today}

\maketitle

\begin{abstract}

We study the properties of  nonlinear Backward Stochastic Differential Equations (BSDEs) driven by a Brownian motion and a  martingale measure associated with a default jump  with intensity process $(\lambda_t)$. We give  a priori estimates 
for these equations and  prove comparison and strict comparison theorems. These results are generalized to drivers involving  a singular  process. The special case of a $\lambda$-linear driver is studied, leading to a representation of the solution of the associated BSDE in terms of a conditional expectation and an adjoint exponential semi-martingale. 
We then apply these results to nonlinear pricing  of European contingent claims in an imperfect  financial market with a totally defaultable risky asset.  The case of claims paying dividends is also studied via a singular process.


 \end{abstract}



\section{Introduction}

The aim of the present paper is to study  BSDEs driven by a Brownian motion and a  martingale measure associated with a default jump process with intensity process $\lambda=(\lambda_t)$. The applications we have in mind are the pricing and hedging issues for 
contingent claims in an imperfect financial market with default. 
 The theory on BSDEs driven by a Brownian motion and a Poisson random measure has been studied  extensively by several authors (we refer e.g. to Barles, Buckdahn and Pardoux \cite{babuc95}, Royer \cite{R}, Quenez and Sulem  \cite{QS1}). 
The present study relies on many arguments which are used in this literature. Nevertheless,   the treatment of a default jump requires some 
specific arguments which are not straightforward and we present here a rigorous analysis of these BSDEs with default jump. To our knowledge, there are  few results on non linear BSDEs with default jump.  The papers \cite{BER} and \cite{ABE} concern only  the existence and the uniqueness of the solution, established under different assumptions (see Remark \ref{remarq} for more details). In our paper, we first provide 
some useful a priori estimates, from which the existence and uniqueness result directly follows. We also give a representation property of the solution when the driver is $\lambda$-linear, as well as comparison and strict comparison theorems in the general case. 
We moreover allow the driver of these equations to have  some singular component, in the sense that the driver may be of the  generalized form
$g(t,\omega,y,z,k)dt + dD_t(\omega), $ where $D$ is a finite variation c\`adl\`ag process with square  integrability conditions. This allows us to treat the case 
of dividends in our financial application.

 The paper is organized as follows:  in Section \ref{sec2},  we present the theory of BSDEs with default jump. More precisely, in Section \ref{setup}, we present the mathematical setup.
 In Section \ref{exiuni}, we state some a priori estimates, from which we derive the existence and the uniqueness of the solution.
In Section \ref{Section23}, we introduce the definition of a {\em $\lambda$-linear} driver, where 
 $\lambda$ refers  to the intensity of the jump process,  which generalizes the notion of a linear driver given in the literature on BSDEs to the case with default jump.
 When the driver is {\em $\lambda$-linear}, we provide an explicit solution of the associated BSDE in terms of a conditional expectation and an adjoint exponential semi-martingale. 
 In Section \ref{compcomp}, we establish a comparison theorem, which holds under an appropriate assumption on the driver.
 We also prove a strict comparison theorem, which requires an additional assumption. Some interesting counterexamples of comparison theorems 
 are given when the assumptions  of these theorems are violated. 

 We then turn to the application in Mathematical Finance in section \ref{sec3}. We consider a financial market with a defaultable risky asset and 
we  study pricing and hedging issues for a European option paying a payoff $\xi$ at the maturity $T$ and intermediate dividends modeled via a singular process $D$. The case of a  perfect market model is first studied via 
 the theory of $\lambda$-linear BSDEs with default jump, while the case of imperfections, expressed via the nonlinearity of the wealth dynamics, 
 is then  treated by the  theory of  nonlinear BSDEs with generalized driver developed in Section \ref{sec2}.  In this  setting, the pricing system is expressed as a nonlinear expectation/evaluation $\mathcal{E}^{g, \cdot}: $ $(\xi,D) \mapsto \mathcal{E}^{^{g,D}}(\xi),$ induced by a nonlinear BSDE with default  jump (solved 
under the primitive probability measure $P$) with generalized driver $g(t, \cdot)dt+dD_t$. Properties of consistency, monotonicity, convexity, non-arbitrage of this pricing system rely on 
the properties of the associated BSDE. As an illustrative example of market imperfections, we consider the case when the seller of the option is  a large investor whose trading strategy may affect the market asset prices and the default intensity.

\section{BSDEs with a default jump}\label{sec2}

\subsection{Probability setup}\label{setup}
Let $(\Omega, \mathcal{G},\mathbb{P})$ be a complete probability space 
 equipped with two stochastic processes:
  a unidimensional standard Brownian motion $W$ and a jump process $N$ defined by 
  $N_t={\bf 1}_{\vartheta\leq t}$ for any $t\in[0,T]$, where $\vartheta$ is a random variable which modelizes a default time. We assume that this default can appear at any time that is $P(\vartheta \geq t)>0$ for any $t\geq 0$. We denote by ${\mathbb G}=\{\mathcal{G}_t, t\geq 0 \}$ the complete natural filtration of $W$ and $N$. We suppose that  $W$ is a ${\mathbb G}$-Brownian motion. 
  
 Let  $(\Lambda_t)$ be the  predictable compensator of the non decreasing process $(N_t)$.
 Note that $(\Lambda_{t \wedge \vartheta})$ is then the predictable compensator of
  $(N_{t \wedge \vartheta} )= (N_t)$. By uniqueness of the predictable compensator, 
  $\Lambda_{t \wedge \vartheta} = \Lambda_t$, $t\geq0$ a.s.
 
  We assume that $\Lambda$ is absolutely continuous w.r.t. Lebesgue's measure, so that there exists a nonnegative process $\lambda$, 
 called the intensity process, such that $\Lambda_t=\int_0^t \lambda_s ds$, $t\geq0$.
  Since $\Lambda_{t \wedge \vartheta} = \Lambda_t$,  $\lambda$ vanishes after $\vartheta$. 

We denote by $M$ the compensated martingale   which satisfies 
\begin{equation}
\label{M}
M_t  = N_t-\int_0^t\lambda_sds\,.
\end{equation}

Let $T >0$ be the finite horizon. We introduce the following sets:
\begin{itemize}
\item ${\cal S}^{2}$ 
is the set of ${\mathbb G}$-adapted RCLL processes $\varphi$ such that $\mathbb{E}[\sup_{0\leq t \leq T} |\varphi_t | ^2] < +\infty$.
\item ${\cal A}^2$  is the set of real-valued non decreasing RCLL adapted
 processes $A$ with $A_0 = 0$ and $\mathbb{E}(A^2_T) < \infty$.
\item ${\mathbb H}^2$  is the set of ${\mathbb  G}$-predictable processes such that
 $
 \| Z\|^2:= \mathbb{E}\Big[\int_0^T|Z_t|^2dt\Big]<\infty \,.
 $
\item  ${\mathbb H}^2_{\lambda}$ is the set of ${\mathbb G}$-predictable processes such that 
$\| U\|_{\lambda}^2:=\mathbb{E}\Big[\int_0^T|U_t|^2\lambda_tdt\Big]<\infty \,.$
\end{itemize}

In the following, ${\mathcal P}$ denotes the ${\mathbb G}$-predictable $\sigma$-algebra on 
$\Omega \times [0,T]$. 

Note that for each $U \in {\mathbb H}^2_{\lambda}$, we have $\| U\|_{\lambda}^2=\mathbb{E}\Big[\int_0^{T\wedge \vartheta}|U_t|^2\lambda_tdt\Big] \,$ because 
the ${\mathbb G}$-intensity $\lambda$ vanishes after $\vartheta$. 
Moreover, we can suppose that for each 
$U$ in ${\mathbb H}^2_{\lambda}$ ($=  L^2(\Omega \times [0,T], {\cal P}, dP   \lambda_t dt)$), its representant, still denoted by 
$U$, vanishes  after $\vartheta$. 
\\
Moreover,  $\mathcal{T}$ is the set of
stopping times $\tau$ such that $\tau \in [0,T]$ a.s.\, and for each $S$ in $\mathcal{T}$, 
   $\mathcal{T}_{S}$ is  the set of
stopping times
$\tau$ such that $S \leq \tau \leq T$ a.s.

We recall the martingale representation theorem in this framework 
(see e.g.  \cite{JYC}):
\begin{lemma}\label{theoreme representation}
Any ${\mathbb  G}$-local martingale $m= (m_t)_{0\leq t \leq T}$ has the representation 
\begin{equation}
\label{equation representation}
m_t  = m_0+\int_0^t z_sdW_s+\int_0^t l_sdM_s \,, \quad\forall\,t\in[0,T] \quad a.s. \,,
\end{equation}
where $z= (z_t)_{0\leq t \leq T}$ and $l= (l_t)_{0\leq t \leq T}$ are 
predictable processes such that the two above stochastic integrals are well defined.
If $m$ is a square integrable martingale, then $z \in {\mathbb H}^2$ and $l \in {\mathbb H}^2_{\lambda}$.
\end{lemma}

We now introduce the following definitions.
\begin{definition}[Driver, $\lambda$-{\em admissible} driver]\label{defd}
A  function $g$
is said to be a {\em driver} if\\
$g:\Omega  \times [0,T] \times \R^3  \rightarrow \R $; 
$(\omega, t,y, z, k) \mapsto  g(\omega, t,y,z,k) $
is $ {\cal P} \otimes {\cal B}(\R^3) 
- $ measurable, and such that
 $g(.,0,0,0) \in {\mathbb H}^2$.
 
A driver $g$ is called a $\lambda$-{\em admissible driver} if moreover there exists a constant $ C \geq 0$ such that 
$dP \otimes dt$-a.s.\,,
for each  $(y, z, k)$, $(y_1, z_1, k_1)$, $(y_2, z_2, k_2)$,
\begin{equation}\label{lip}
|g( t, y, z_1, k_1) - g( t, y, z_2, k_2)| \leq
C (|y_1 - y_2|+ |z_1 - z_2| +   \sqrt \lambda_t |k_1 - k_2 |). 
\end{equation}
The positive real $C$ is called the $\lambda$-{\em constant} associated with driver $g$.
\end{definition}

Note that condition \eqref{lip} implies that  for each $\,t > \vartheta$, since $\lambda_t=0$,
$g$ does depend on $k$. In other terms, for each $(y,z,k)$, we have: 
 $$g(t,y,z,k)= g(t,y,z,0), \quad t > \vartheta \quad dP \otimes dt-{\rm 
 a.s.}$$

 
\begin{definition}[BSDE with default jump]
Let  $g$ be a $\lambda$-{\em admissible} driver, let $\xi \in {L}^2({\cal G_T})$. 
A process
 $(Y, Z, K)$  in $ \mathcal{S}^2 \times {\mathbb H}^2 \times  {\mathbb H}^2_{\lambda}$ is said to be a solution of the BSDE with default jump associated with terminal time $T$, driver $g$ and terminal condition $\xi$ if it satisfies:
\begin{equation}\label{BSDE}
-dY_t = g(t,Y_t, Z_t,K_t ) dt  -  Z_t dW_t - K_t dM_t; \quad
Y_T=\xi.
\end{equation}
\end{definition}

 \subsection{First properties of BSDEs with a default jump}\label{exiuni}

For $\beta >0$,   $\phi \in \H^{2} $, and $l \in \H_\lambda^{2}$, we introduce the norms $\| \phi \|_{\beta}^2 := E[\int_0^T e^{\beta s} \phi_s^2 ds], $ and 
 $\| k \|_{\lambda,\beta}^2 := E[\int_0^T e^{\beta s} k_s^2 \lambda_s \, ds] $. We first show some a priori estimates for BSDEs with a default jump, from which we derive the existence and uniqueness of the solution.
 \subsubsection{A priori estimates for BSDEs with a default jump}
\begin{proposition} \label{est}
Let $\xi^1$, $\xi^2$ $\in L^2({\cal G}_T)$. Let $g^1$ and $g^2$ be two  $\lambda$-{\em admissible} drivers.
 For $i=1,2$, let $(Y^i, Z^i ,K^i)$  be  a solution of the BSDE associated with  
terminal time $T$, driver $g^i$  and terminal conditions $\xi^i$.
Let $\bar \xi := \xi^1 - \xi^2$ and $\bar g(s): = g^1(s, Y^2_s, Z^2_s, K_s^2) - g^2(s, Y^2_s, Z^2_s, K_s^2)$. For $s$ in $[0,T]$, denote $\bar Y_s := Y^1_s - Y^2_s, \,\,\, \bar Z_s := Z^1_s - Z^2_s$,  $\bar K_s := K^1_s - K^2_s $. \\
Let $ \eta, \beta >0 $ be such that 
 $\beta \geq \frac{3}{\eta} +2C $ 
and $\eta \leq \frac{1}{C^2}$.
For each $t \in [0,T]$, we then have
\begin{equation}\label{estimateY}
e^{\beta  t} (\bar Y_t)   ^2 \leq  \EB [ e^{\beta  T} \bar \xi \,^2 \mid 
{\cal G}_t ] +\ \eta \,{\mathbb E}[ \int_t^T e^{\beta  s} \bar g(s)^2  ds \mid 
{\cal G}_t ] \;\; \text{ \rm a .s.}\, 
\end{equation}
Moreover, 
\begin{equation}\label{estimateY2}
\|\bar Y \|_\beta^2 \leq  T [e^{\beta T} \EB[\bar \xi \,^2] + \eta
\|\bar g \|_\beta^2].
\end{equation}
If $\eta < \frac{1}{C^2}$, we  have 
\begin{equation}\label{estimateZ}
\|\bar Z \|_\beta^2 + \|\bar K \|_{\lambda,\beta}^2
\leq \frac{1}{1 - \eta C^2} [e^{\beta T} \EB[\bar \xi \,^2] + \eta \|\bar g \|_\beta^2].
\end{equation}
\end{proposition}

\begin{proof}  
By It\^o's formula applied  to the semimartingale $e^{\beta s} \bar Y_s$ 
between $t$ and $T$, we get 
\begin{align}
e^{\beta t} \bar Y_t ^2  & + \beta \int_t^T e^{\beta s} \bar Y_s^2 ds +
 \int_t^T e^{\beta s} \bar Z_s^2 ds + \int_t^T e^{\beta s}  \bar K_s^2 \lambda_s ds
 \nonumber \\  
 &=e^{\beta T} \bar Y_T ^2 +  2 \int_t^T e^{\beta s} \bar Y_s (g^1 (s, Y^1_s, Z^1_s, K^1_s) - g^2 (s, Y^2_s, Z^2_s, K^2_s)) ds  \nonumber \\
 &\quad - 2 \int_t^T e^{\beta s} \bar Y_s \bar Z_s dW_s
  -  \int_t^T e^{\beta s} (2  \bar Y_{s^-} \bar K_s +  \bar K_s^2) d M_s.  \label{russ}
\end{align}
Taking the conditional expectation given ${\cal G}_t$, we obtain
\begin{align}
e^ {\beta t} & \bar Y_t^2 
 +  E \left[\beta \int_t^T e^{\beta s}  \bar Y_s^2 ds + \int_t^T e^{\beta s} ( \bar Z_s^2 + \bar K_s^2 \lambda_s) ds \mid {\cal G}_t \right] \nonumber \\
& \leq  E  \left[   e^{\beta T} \bar Y_T ^2     \mid {\cal G}_t\right]  + 2  E  \left[    \int_t^T e^{\beta s} \bar Y_s (g^1 (s, Y^1_s, Z^1_s, K^1_s) - g^2 (s, Y^2_s, Z^2_s, K^2_s)) ds   \mid {\cal G}_t\right].
\end{align}
Now, $ g^1(s,Y^1_s, Z^1_s, K^1_s) - g^2(s,Y^2_s,Z^2_s,K^2_s)=  g^1(s,Y^1_s, Z^1_s, K^1_s) - g^1(s,Y^2_s,Z^2_s,K^2_s) + \bar g_s$.\\
Since $g^1$ satisfies condition \eqref{lip}, we derive that
$$
|g^1(s,Y^1_s, Z^1_s, K^1_s) - g^1(s,Y^2_s,Z^2_s,K^2_s)| 
  \leq C|\bar Y_s| + C|\bar Z_s| + C |\bar K_s| \sqrt \lambda_s + |\bar g_s|.
$$
Note that,  for all non negative numbers $\lambda$, $y$, 
$z$, $k$, $g$ and $\varepsilon >0$, we have\\
$ 2y (Cz + Ck \sqrt  \lambda  + g) \leq \frac{ y^2}{\varepsilon^2}+ \varepsilon^2(Cz+ Ck  \sqrt \lambda   + g)^2 \leq \frac{ y^2}{\varepsilon^2} +  3 \varepsilon^2(C^2 y^2+ C^2 k^2\lambda +g^2)$. Hence, 
\begin{align}\label{eq2a}
e^ {\beta t}  \bar Y_t^2  +  {\mathbb E} \left[\beta \int_t^T e^{\beta s}  \bar Y_s^2 ds +  \int_t^T e^{\beta s} ( \bar Z_s^2 +
 \bar K_s^2\lambda_s) ds \mid {\cal G}_t \right]    \leq \EB  \left[   e^{\beta T} \bar Y_T ^2     \mid {\cal G}_t\right]  \nonumber \\
 + {\mathbb E} \left[ (2C+\frac{ 1}{\varepsilon^2}) \int_t^T  e^{ \beta s}  \bar Y_s^2 ds + 3C^2 \varepsilon^2 \int_t^T e^{\beta s} ( \bar Z_s^2 + \bar K_s^2\lambda_s)ds
  + 3 \varepsilon^2   \int_t^T e^{\beta s}  \bar g_s^2 ds \mid {\cal G}_t \right].
\end{align}
Let us make the change of variable $\eta = 3 \epsilon^2$. Then, for each  $\beta,  \eta>0$  chosen as in the proposition,  these inequalities lead to
\eqref{estimateY}. 
By integrating  \eqref{estimateY}, we obtain \eqref{estimateY2}. Using  inequality \eqref{eq2a}, we derive \eqref{estimateZ}.
\end{proof}

\begin{remark} \label{AA29}
 By classical results on the norms of semimartingales, one similarly shows that
$
\|  \bar Y \|_{S^2}$ $ \leq$ $ K \left( \EB [\bar \xi \,^2] + \|\bar g \|_{\H^2}\right)
$, where $K$ is a positive constant only depending on $T$ and $C$.
\end{remark}
 \subsubsection{Existence and uniqueness result for BSDEs with a default jump}
  
 \begin{proposition}\label{existence} Let  $g$ be a $\lambda$-admissible driver, let $\xi \in {L}^2({\cal G_T})$. 
 There exists an unique solution   $(Y, Z, K)$  in $ \mathcal{S}^2 \times {\mathbb H}^2 \times  {\mathbb H}^2_{\lambda}$ of the 
 BSDE \eqref{BSDE}.
\end{proposition}

\begin{remark}\label{remarq}
This result generalizes the existence and uniqueness result obtained in \cite{BER} under stronger assumptions.
Indeed, suppose that $\xi$ is  $\mathcal{G}_{\vartheta \vee T}$-measurable (as in \cite{BER}) and that $g$ is replaced by $g {\bf 1}_{t \leq \tau}$ (which is a $\lambda$-admissible driver), then the solution $(Y, Z, K)$ of the associated
 BSDE \eqref{BSDE} is equal to the solution of the BSDE  with random terminal time $\vartheta$, driver $g$ and terminal condition $\xi$, considered in  \cite{BER}. Note that the boundedness assumption made on the default intensity process 
$(\lambda_t)$ in \cite{BER} is not necessary to ensure this result.
\end{remark}

\dproof  We show this  result  by using the  a priori estimates given in Proposition \ref{est}. The arguments are classical and a short proof is given for completeness. Let us first consider the case when the driver $g(t)$ does not depend on the solution. 
By using the representation property of ${\cal G}$-martingales (Lemma \ref{theoreme representation}) together with classical computations, one can show that there exists a unique solution of the BSDE \eqref{BSDE} associated with terminal condition $\xi \in L^2(\Fc_T)$ and
driver process $g(t)$ $\in$ $\H^2$. 
Let us turn to the case with a general $\lambda$-admissible driver $g(t,y,z,k)$.
Denote by $\H_\beta^2$ the space $\H^2 \times \H^2 \times \H^2_\lambda$ equipped with the norm 
$\| Y, Z, K \|_\beta^2 := \| Y \|_{\beta}^2 +  \| Z \|_{\beta}^2  +  \| K \|_{\lambda,\beta}^2 $.
We define a mapping $\Phi$ from $\H_\beta^2$ into
itself as follows. Given $(U, V, L) \in \H_\beta^2$, let $(Y, Z,K) = \Phi (U, V, L)$
be the 
the solution of the BSDE  associated with driver $g^1(s) := 
g(s, U_s , V_s, L_s)$.  \\
%
Let us prove that the mapping $\Phi$ is a contraction from ${\H}_\beta^2$ into ${\H}_\beta^2$.
%
Let $(U', V', L' )$ be another element of ${\H}_\beta^2$  and let $(Y', Z',k') := 
\Phi (U', V', L')$, that is,
the solution of the RBSDE  associated with driver process $g(s, U'_s , V'_s, L'_s)$. \\
Set $\bar{U} = U - U'$, $\bar{V} = V - V'$, $\bar{L} = L - L'$, $\bar{Y} = Y - Y'$, $\bar Z = Z-Z'$,  $\bar{K} = K - K'$. \\
Let $\Delta g_\cdot :=  g(\cdot, U, V, L) - g(\cdot, U', V', L')$.
Using the estimates \eqref{estimateY2} and 
\eqref{estimateZ}  with $\lambda$-constant equal to  $0$ (since the driver $g^1$ does not depend on the solution), we derive that for all $\eta, \beta >0$ such that $\beta \geq \frac{3}{\eta}$, we have
$$
\|\bar{Y}\|_\beta^2 + \|\bar{Z}\|_\beta^2 + \|\bar{K}\|_{\lambda, \beta}^2  \leq 
\eta (T+1) \|\Delta g \|_\beta^2.$$
 Since the driver $g$ is $\lambda$-{\em admissible} with $\lambda$-constant $C$, we get
$$
\|\bar{Y}\|_\beta^2 + \|\bar{Z}\|_\beta^2 + \|\bar{K}\|_{\lambda, \beta}^2  \leq
\eta (T+1) 3C^2  (\| \bar{U}\|_\beta^2 + \| \bar{V}\|_\beta^2 + \| \bar{L}\|_{\lambda, \beta}^2),
$$
for all $\eta, \beta >0$ with $\beta \geq \frac{3}{\eta}$.
Choosing 
$\eta = \frac{1} {(T+1) 6C^2 }$ and $\beta = \frac{3}{\eta}$,
we derive that  
$ \| (\overline{Y}, \overline{Z}, \overline{K}) \|_\beta^2 \leq \frac{1}{2} 
\| (\overline{U}, \overline{V}, \overline{K}) \|_\beta^2. $
Hence, for $\beta = 18 (T+1) C^2$, $\Phi$ is a contraction from ${\H}_\beta^2$ into ${\H}_\beta^2$ and thus admits a unique fixed point $(Y, Z, K)$ in the Banach space ${\H}_\beta^2$, which 
 is the solution of BSDE~\eqref{BSDE}.
\fproof

By similar arguments, we have the following generalized result.
\begin{proposition} \label{existencegeneral}[BSDEs with default jump and ``generalized driver"] Let  $g$ be a $\lambda$-admissible driver, let $\xi \in {L}^2({\cal G_T})$ and 
let $D$ be a finite variational RCLL adapted process with square integrable total variation process. There exists an unique solution  $(Y, Z, K)$ (also denoted by $(Y^D(T, \xi), Z^D(T, \xi), K^D(T, \xi))$)  in $ \mathcal{S}^2 \times {\mathbb H}^2 \times  {\mathbb H}^2_{\lambda}$ of 
 the BSDE associated with ``generalized driver" $g(t, \cdot)dt +dD_t$ and terminal condition $\xi$, that is
\begin{equation}\label{BSDEgeneral}
-dY_t = g(t,Y_t, Z_t,K_t ) dt + dD_t -  Z_t dW_t - K_t dM_t; \quad
Y_T=\xi.
\end{equation}
\end{proposition}

\begin{remark}
Let $D$ be a finite variational RCLL adapted process. Its associated total variation process is square integrable if and only if $D$ can be decomposed as follows: $D=A-A',$ with $A$, $A' \in \mathcal{A}^2$.
\end{remark}

\subsection{$\lambda$-{\em linear} BSDEs with default jump}\label{Section23}

We introduce the notion of $\lambda$-{\em linear} BSDEs in our framework with default jump. 

\begin{definition}[$\lambda$-{\em linear} driver] \label{deflinear}
A driver $g$ is called $\lambda$-{\em linear} if it is of the form: 
\begin{equation}\label{ll}
g(t,y,z,k)= \varphi_t + \delta_t y+ \beta_t z+ \gamma_t \, k\, \lambda_t,
\end{equation}
where
 $(\varphi_t) \in {\mathbb H}^2$, and where $(\delta_t)$,  $(\beta_t)$ and $(\gamma_t)$ are $\R$-valued predictable processes such that $(\delta_t)$,  $(\beta_t)$ and $(\gamma_t
 \sqrt{ \lambda_t})$ are bounded. 
 \end{definition}
 
 \begin{remark}
 Note that a driver $g$ is $\lambda$-{\em linear} if and only if it is of the form: 
 \begin{equation}\label{llbis}
g(t,y,z,k)= \varphi_t + \delta_t y+ \beta_t z+ \nu_t \, k\, \sqrt{\lambda_t},
\end{equation}
where
 $(\varphi_t) \in {\mathbb H}^2$, and where $(\delta_t)$,  $(\beta_t)$ and $(\nu_t)$ are bounded $\R$-valued predictable processes.
 
 From this observation, it follows that a $\lambda$-{\em linear} driver is $\lambda$-{\em admissible}. 
   \end{remark}

We will now prove that the solution of a $\lambda$-{\em linear} BSDE, that is the solution of BSDE \eqref{BSDE} associated with a 
$\lambda$-{\em linear} driver,  can be written as a conditional expectation via an exponential semimartingale. We first show a preliminary result. 

%
\begin{proposition}\label{roro}  
Let $(\beta_s)$ and $(\gamma_s)$ be  two real-valued ${\mathbb G}$-predictable processes such that 
the random variable $\int_0^T  (\beta_r ^2 + \gamma_r ^2  \lambda _r)\,dr$ is bounded.
 Let $(\zeta_s)$ be the process satisfying the forward SDE:\\ 
$\quad d \zeta_{s}  = \zeta_{s^-} (\beta_s d W_s +\gamma_s dM_s)$ with $\zeta_0 =1$.\\
 The process $(\zeta_s)$ satisfies
the so-called Dol\'eans-Dade formula, that is
\begin{align*}
\zeta_s =  \exp\{ \int_0^s \beta_r d W_r  - \frac{1}{2} \int_0^s \beta^2_r dr \}\exp\{- \int_0^s  \gamma_r \lambda_rdr \}(1 + \gamma_{\vartheta}{\bf 1}_{ \{ s \geq \vartheta \}}).
\end{align*}
For each $T >0$, 
the process $(\zeta_s)_{0\leq s\leq T}$ is a  martingale and satisfies 
$ {\mathbb E}[\sup_{0 \leq s \leq T}\zeta_s^p] < + \infty$, for all $p \geq 2$.\\
Moreover, if $\gamma_{\vartheta} \geq -1$ (resp. $> -1$) a.s.\,, then $\zeta_s \geq 0$ (resp. $>0$) a.s. for each 
$s \in [0,T]$. 

\end{proposition}

\dproof By definition, the process  $(\zeta_s)$ is a local martingale. Let $T>0$.
Let us show that $ {\mathbb E}[\sup_{0 \leq s \leq T}\zeta_s^2] < + \infty.$
By It\^o's formula applied to $\zeta_s^2$, we get 
$d\zeta_s^2 = 2 \zeta_{s^-} d\zeta_s + d[\zeta,\zeta]_s$. We have
$$d[\zeta,\zeta]_s= \zeta_{s^-}^2 \beta_s^2 ds + \zeta_{s^-}^2 \gamma_s^2 dN_s.$$  Using \eqref{M}, we thus 
derive that
$$d\zeta_s^2 = \zeta_{s^-} ^2 [ 2 \beta_s dW_s + (2 \gamma_s    + \gamma_s^2) dM_s +   (\beta_s ^2 +\gamma_s^2 \lambda_s )ds ].$$
It follows that $\zeta^2$ is an exponential semimartingale which can be written: 
\begin{equation}\label{non}
\zeta_s^2 =   \eta_s \exp \{\int_0^s  (\beta_r ^2 + \gamma_r ^2  \lambda _r)\,dr \},
  \end{equation}
where $\eta$ is the exponential local martingale satisfying 
$$d \eta_{s}  = \eta_{s^-} [ 2 \beta_s dW_s+
(2 \gamma_s    + \gamma_s^2) dM_s ],$$ with $\eta_0 =1$.
By equality \eqref{non}, the local martingale $\eta$ is non negative. Hence, it is a
supermartingale, which yields that ${\mathbb E}[ \eta_T] \leq 1$. 
Now, by assumption, 
$\int_0^T  (\beta_r ^2 + \gamma_r ^2  \lambda _r)\,dr$ is bounded.
 By \eqref{non}, it follows that
$${\mathbb E}[\zeta_T^2] \leq {\mathbb E} [\eta_T]\,K \leq K,$$ where $K$ is a positive constant. By martingale inequalities, 
we derive that $ {\mathbb E}[\sup_{0 \leq s \leq T}\zeta_s^2] < + \infty.$ Hence, 
the process $(\zeta_s)_{0\leq s\leq T}$ is a  martingale. 
By an induction argument as in the proof of Proposition A.1 in \cite{QS1}, one can prove that the integrability property of 
$\zeta$ holds for all integer $p\geq 2$.

The last assertion follows from the Dol\'eans-Dade formula.
 \fproof
 
 \begin{remark}\label{moinszero}
The inequality $\gamma_{\vartheta} \geq -1$  a.s.\, is equivalent to the inequality 
$\gamma_t \geq -1$, $dt\otimes dP \text{ a.s.}$. Indeed, we have 
${\mathbb E} [{\bf 1}_{\gamma_{\vartheta} < -1}]$$ ={\mathbb E} [\int_0^{+ \infty}
{\bf 1}_{\gamma_{r} < -1}dN_r]$ $= {\mathbb E} [\int_0^{+ \infty}
{\bf 1}_{\gamma_{r} < -1}\lambda_r dr]$, because the process $(\int_0^t \lambda_r dr)$ is the ${\mathbb G}$-predictable compensator of the default jump process $N$.

\end{remark}
%
%

%
%


\begin{theorem}[Representation of the solution of a $\lambda$-{\em linear} BSDE]\label{linear} 
Let $\xi$ $\in$ $L^2({\cal G}_T)$.  Let $g$ be a $\lambda$-{\em linear} driver of the form \eqref{ll}.
Let $(Y, Z, K)$ be the solution in $S^{2} \times \H^{2} \times \H_{\lambda}^{2}$ of the BSDE associated with driver $g$ and terminal condition $\xi$, that is
\begin{align}\label{lin}
- dY_t = (\varphi_t + \delta_t Y_t + \beta_t Z_t + \gamma_t K_t \lambda_t )dt  -  Z_t  dW_t - K_t  dM_t; \;\; Y_T = \xi.
\end{align}
For each $t$ $\in$ $[0,T]$, let $(\Gamma_{t,s})_{s \in [t,T]}$ (called the {\em adjoint process}) 
be the unique solution of the following 
forward SDE
\begin{equation}\label{eq4b}
d \Gamma_{t,s}  = \displaystyle \Gamma_{t,s^-} \left[ \delta_s ds + \beta_s d W_s +
 \gamma_s dM_s\right] ; \;\; 
\Gamma_{t,t}  = 1. 
\end{equation}
The process $(Y_t)$  satisfies
\begin{equation} \label{eqqun}
 Y_t \,= \,{\mathbb E} \,[ \, \Gamma_{t,T} \,\,\xi  + \int_t^T \Gamma_{t,s}\, \varphi_s 
 ds \mid {\cal G}_t \, ], \quad 0 \leq t \leq T,\quad {\rm a.s.}\,
\end{equation}
\end{theorem}
\begin{remark}\label{moinsun}

The process  $(\Gamma_{t,s})_{s \in [t,T]}$, defined by \eqref{eq4b}, satisfies
 \begin{align*} 
  \Gamma_{t,s} =  e^{ \int_t^s \delta_s ds} \exp\{ \int_t^s \beta_r d W_r  - \frac{1}{2} \int_t^s \beta^2_r dr \}\exp\{- \int_t^s  \gamma_r \lambda_rdr \}(1 + \gamma_{\vartheta}{\bf 1}_{ \{ s \geq \vartheta >t\}}).
  \end{align*}
  
Note that the process 
$(e^{ \int_t^s \delta_s ds})_{t \leq s \leq T}$ is positive and bounded since $\delta$ is  bounded.
 Using Proposition \ref{roro}, since $\beta$ and $\gamma
 \sqrt{ \lambda}$ are bounded, we derive that $ {\mathbb E}[\sup_{t \leq s \leq T}\Gamma_{t,s}^2] < + \infty.$ Moreover,  if $\gamma_{\vartheta} \geq -1$ (resp. $> -1$) a.s.\,, then $\Gamma_{t,s} \geq 0$ (resp. $>0$) a.s. for each 
$s \in [t,T]$.
\end{remark}

\dproof Fix $t$ $\in$ $[0,T]$. 
By applying  the It\^o product formula to $Y_s \Gamma_{t,s}$, we get
\begin{align*}
-d(Y_s \Gamma_s) & =- Y_{s^-} d\Gamma_s - \Gamma_{s^-}dY_s - d [ Y, \Gamma]_s\\
& = -Y_s  \Gamma_s \delta_s ds + \Gamma_{s} \left[ \varphi_s + \delta_s Y_s + \beta_s Z_s + 
\gamma_s K_s \lambda_s  \right]ds \\
  & \quad - \beta_s Z_s \Gamma_s ds -  \Gamma_s  \gamma_s K_s \lambda_s ds - \Gamma_s 
  (Y_s  \beta_s+ Z_s) dW_s 
  - \Gamma_{s^-} 
  [K_s (1 +\gamma_s )+Y_{s^-} \gamma_s] dM_s,
 \end{align*}
(where $\Gamma_{t,s}$ is denoted by $\Gamma_s)$.
Setting $$dm_s =  -\Gamma_{t,s} (Y_s  \beta_s+ Z_s) dW_s 
  - \Gamma_{t,s^-}  [K_s(1 +\gamma_s) +Y_{s^-} \gamma_s]dM_s,$$ we get  
  $-d(Y_s \Gamma_s)=\Gamma_s \varphi_s ds - d m_s$.
   By integrating between $t$ and $T$, we obtain
  \begin{equation}\label{xx}
  Y_t=\xi \Gamma_{t,T}  +  \int_t^T \Gamma_{t,s} \varphi_s ds   - (m_T -m_t)  \quad {\rm a.s.}
  \end{equation}
By Remark \ref{moinsun}, we have  $(\Gamma_{t,s})_{ t\leq s\leq T}$ $\in$ $S^{2}$. 
Moreover, $Y$ $\in$ $S^{2}$, $Z$ $\in$ $\H^{2}$, $K$ $\in$ $\H^{2}_{\lambda}$, and $\beta$ and $\gamma$ are bounded. It follows that 
the local martingale $m= (m_s)_{t\leq s\leq T}$ is a martingale. Hence, by taking the conditional expectation in equality (\ref{xx}), we get equality (\ref{eqqun}). 
\fproof

By similar arguments, we have the  generalized representation result.
\begin{proposition} \label{representationgeneral} Let $\xi \in {L}^2({\cal G_T})$ and 
let $D$ be a finite variational RCLL adapted process with square integrable total variation process. Let $(\delta_t)$,  $(\beta_t)$ and $(\gamma_t)$ be $\R$-valued predictable processes such that $(\delta_t)$,  $(\beta_t)$ and $(\gamma_t
 \sqrt{ \lambda_t})$ are bounded. 

Let  $(Y,Z,K)$  be the solution of 
 the BSDE: 
$$
- dY_t = ( \delta_t Y_t + \beta_t Z_t + \gamma_t K_t \lambda_t )dt + dD_t  -  Z_t  dW_t - K_t  dM_t; \;\; Y_T = \xi.
$$
For each $t$ $\in$ $[0,T]$, 
we have
\begin{equation} \label{eqgeneral}
 Y_t \,= \,{\mathbb E} \,[ \, \Gamma_{t,T} \,\,\xi  + \int_t^T \Gamma_{t,s}\, dD_s \mid {\cal G}_t \, ] \quad {\rm a.s.}\,,
\end{equation}
where $(\Gamma_{t,s})_{s \in [t,T]}$ is the adjoint process defined by \eqref{eq4b}.
\end{proposition}

\subsection{Comparison theorems for BSDEs with a default jump}\label{compcomp}

We give here a comparison theorem and a strict comparison result for BSDEs with a default jump under additional assumptions  on the driver.

\begin{theorem}[Comparison theorems for BSDEs with default jump]\label{comparisonBSDE}
Let $\xi_1$ and $\xi_2$ $\in$ $L^2({\cal G}_T)$.
Let $g_1$ and $g_2$ be two $\lambda$-{\em admissible} drivers. For $i=1,2$, let $(Y^i, Z^i, K^i)$ be a solution in $S^{2} \times \H^{2} \times \H_{\lambda}^{2}$ of the BSDE
\begin{equation}\label{eq7}
-dY^i_t  = \displaystyle   g_i (t, Y^i_t, Z^i_t, K^i_t) dt - Z^i_t dW_t - K^i_t dM_t; \quad Y^i_T  = \xi_i.
\end{equation}
(i) {\em (Comparison theorem)}. Assume that there exists a  predictable process $(\gamma_t)$ with 
 \begin{equation}\label{robis}
(\gamma_t
 \sqrt{ \lambda_t}) \;\text { bounded} \;\;   \text{ and } \quad\gamma_t \geq -1, \quad  dt\otimes dP \text{ a.s.}\,
\end{equation}
such that
\begin{equation}\label{autre}
 g_1(t, Y^2_t, Z^2_t, K^1_t) - g_1(t,Y^2_t,Z^2_t,K^2_t) \geq  \gamma_t ( K^1_t- K^2_t) \lambda_t , \;\; t \in [0,T],\; \; dt\otimes dP \text{ a.s.}
\end{equation}
Suppose also that
\begin{equation}\label{eq99}
\xi_1 \geq \xi_2 \text{ a.s. }\quad {\rm and} \quad g_1 (t, Y^2_t, Z^2_t, K^2_t) \geq g_2 (t, Y^2_t, Z^2_t, K^2_t), \;\; t \in [0,T],\; \; dt\otimes dP \text{ a.s.}
\end{equation}
We then have $Y^1_t \geq Y^2_t$ for all $t \in [0,T]$.\\
(ii) {\em (Strict Comparison Theorem)}. Suppose moreover that the second inequality in \eqref{robis} is strict,  that is $\gamma_t > -1$.
 If 
$Y_{t_0}=0$, for some $t_0 \in [0,T]$, then the inequalities in \eqref{eq99} are equalities.
\end{theorem}

\dproof 
Setting $\bar{Y}_s = Y^1_s - Y^2_s$ ; $\bar{Z}_s = Z^1_s - Z^2_s$ ; $\bar{K}_s = K^1_s - K^2_s,$ we have
  $$
  -d \bar{Y}_s  \displaystyle =  h_s ds - \bar{Z}_s dW_s - \bar{K}_s dM_s; \quad
   \bar{Y}_s  = \xi_1 - \xi_2,
   $$
where $h_s:=g_1 (s, Y^1_s, Z^1_s, K^1_s) - g_2(s, Y^2_s, Z^2_s, K^2_s)$. \\
Set  $\delta_s  := (g_1(s,Y^1_{s^-}, Z^1_s, K^1_s) - g_1(s, Y^2_{s^-}, Z^1_s, K^1_s))/ \bar{Y}_s$ if 
$\bar{Y}_s \neq 0$, and $0$ otherwise.\\
Set  $\beta_s  := (g_1(s,Y^2_{s^-}, Z^1_s, K^1_s) - g_1(s, Y^2_{s^-}, Z^2_s, K^1_s))/ \bar{Z}_s$ if 
$\bar{Z}_s \neq 0$, and $0$ otherwise.\\
   By classical linearization techniques, we obtain
\begin{equation*}\label{in}
h_s  = \delta_s \bar{Y}_s + \beta_s \bar{Z}_s 
  + g_1(s, Y^2_s, Z^2_s, K^1_s) - g_1(s,Y^2_s,Z^2_s,K^2_s)    + \varphi_s,
\end{equation*}
where $\varphi_s := g_1(s,Y^2_{s^-}, Z^2_s, K^2_s) - g_2(s, Y^2_{s^-}, Z^2_s, K^2_s)$.
Using the assumption \eqref{autre}, we get
\begin{equation}\label{eq2}
  \,h_s  \geq  \delta_s \bar{Y}_s + \beta_s \bar{Z}_s + \gamma_s\, \,\bar{K}_s \lambda_s 
  +\varphi_s \,\, \quad ds \otimes dP-{\rm a.s.}
  \end{equation}

 Since $g_1$ satisfies condition \eqref{lip}, the predictable processes $\delta$ and $\beta$ are bounded. 
Fix $t$ $\in$ $[0,T]$. Let $\Gamma_{t,.}$ be the process defined by (\ref{eq4b}). 
Since $\delta$, $\beta$ and $\gamma  \sqrt{ \lambda}$ are bounded, it follows from Remark \ref{moinsun}  that 
$\Gamma_{t,.}$ $\in$ $S^2$. 
Also, since $\gamma_s  \geq -1$, we have $\Gamma_{t,.} \geq 0$ a.s.\,


By  It\^o's formula and similar computations as in the
 proof of Theorem \ref{linear}, we derive that
\begin{align*}
-d(\bar Y_s \Gamma_{t,s}) & =    \Gamma_{t,s}(h_s  - \delta_s \bar{Y}_s - \beta_s \bar{Z}_s -
 \gamma_s \,\bar{K}_s \,\lambda_s)\,ds - dm_s,
 \end{align*}
 where $m$ is a martingale (because $\Gamma_{t,.}$ $\in$ $S^2$, $\bar Y$ $\in$ $S^{2}$, $\bar Z$ $\in$ $\H^{2}$, $\bar K$ $\in$ $\H^{2}_{\lambda}$ and $\beta$, 
 $\gamma  \sqrt{ \lambda}$ are bounded). Using the inequality (\ref{eq2}) together with the non negativity of $\Gamma$, we thus get
$
-d(\bar Y_s \Gamma_{t,s})
  \geq   \Gamma_{t,s} \varphi_s ds - d m_s 
 $.
 By integrating between $t$ and $T$  and by taking the conditional expectation, we obtain 
%
\begin{equation}\label{comp}
  \bar{Y}_t \,\geq \,{\mathbb E}\, [\, \Gamma_{t,T} \,(\xi_1 - \xi_2) + \int_t^T 
  \Gamma_{t,s}\, \varphi_s ds \mid {\cal G}_t ],\;\; 0 \leq t \leq T,\quad {\rm a.s.}\,
\end{equation}
By assumption \eqref{eq99}, $\varphi_s \geq 0$   and $\xi_1 - \xi_2$ $\geq 0$, which, together with the non negativity of 
$\Gamma_{t,T}$, implies that $
\bar Y= Y^1 - Y^2 \geq 0$. The assertion (i) thus follows. Suppose now that $\gamma_t > -1$. By Remark \ref{moinsun},  $\Gamma_{t,T}>0$ a.s.\, The assertion (ii) thus follows from \eqref{comp}.
\fproof

We give here some counter-examples related to the comparison theorems for BSDEs with a default jump.

\begin{remark}
Let us give an example which shows that in the case where  assumption \eqref{robis} is violated, that is when $\gamma$ takes values $<-1$ with positive measure, then, even if the terminal condition is nonnegative, the solution $Y$ of the linear BSDE with default jump may take strictly negative values. Hence, in this case, the comparison theorem does not hold.
Suppose that the process $\lambda$ is bounded. Let $g$ be a $\lambda$-linear driver of the particular form
\begin{equation}
g(t,\omega,k)=\gamma k \lambda_t(\omega),
\end{equation} 
where  $\gamma$ is a real constant (this corresponds to the driver of the $\lambda$-linear BSDE \eqref{lin} with $\delta_s= \beta_s=\varphi_s=0$ and 
$\gamma_s= \gamma$).
At  terminal time $T$, the associated {\rm adjoint} process $\Gamma_{0,s}$ satisfies (see \eqref{eq4b} and Remark \ref{moinsun}) : 
  \begin{equation}\label{eqq}
  \Gamma_{0,T} = \exp\{- \int_0^T  \gamma \lambda_rdr \}(1 + \gamma {\bf 1}_{ \{ T \geq \vartheta  \}})=  \exp\{- \int_0^T  \gamma \lambda_rdr \}(1 + \gamma N_T),
  \end{equation}
where the second equality follows from  the definition of the default jump process $N$.\\
Let $Y$ be the solution of the BSDE associated with driver $g$ and terminal condition 
$$\xi:=N_T.$$
 The representation property of linear BSDEs with default jump (see \eqref{eqqun}) gives $$Y_0=\mathbb{E}[\Gamma_{0,T}  \xi]=\mathbb{E}[\Gamma_{0,T}  N_T ].$$
Hence, by $\eqref{eqq}$, we get 
\begin{align}\label{eqrefff}
Y_0=\mathbb{E}[\Gamma_{0,T}N_T]=\mathbb{E}[e^{-\gamma\int_0^T \lambda_sds}(1+\gamma N_T) N_T]=(1+\gamma)\mathbb{E}[e^{-\gamma\int_0^T \lambda_sds}N_T],
\end{align}
where for the last equality we have used the fact that $N_T=N_T^2$. \\
Equation \eqref{eqrefff}  shows that when $\gamma<-1$, we have $Y_0<0$ although $\xi \geq 0$ a.s.

This example also gives a counter-example for   the strict comparison theorem by taking $\gamma=-1$.  Indeed,  in this case, the relation $\eqref{eqrefff}$ at time $0$ yields that $Y_0=0$. Now,  we have 
\begin{align}
\mathbb{E}[\xi]= \mathbb{E}[N_T]=1-P(\vartheta > T). 
\end{align}
Hence, under the additional assumption $P(\vartheta>T)<1$, we get $\mathbb{E}[\xi]>0$, which implies that 
 $P(\xi>0)>0$, even though $Y_0=0$.
 
\end{remark}


\begin{proposition}[Comparison theorems for BSDEs with ``generalized driver"]\label{comparisongeneral} Let $\xi_1$ and $\xi_2$ $\in$ $L^2({\cal G}_T)$.
Let $g_1$ and $g_2$ be two $\lambda$-{\em admissible} drivers.
 Let $D^1$ and $D^2 $  be  finite variational RCLL adapted processes with square integrable total variation. Let $(Y^i, Z^i, K^i)$ be a solution in $S^{2} \times \H^{2} \times \H_{\lambda}^{2}$ of the BSDE 
\begin{equation*}
-dY^i_t  = \displaystyle   g_i (t, Y^i_t, Z^i_t, K^i_t) dt +dD^i_t- Z^i_t dW_t - K^i_t dM_t; \quad Y^i_T  = \xi_i.
\end{equation*}
(i) {\em (Comparison theorem)}. Assume that there exists a predictable process $(\gamma_t)$ satisfying \eqref{autre}
with \eqref{robis} and that 
\eqref{eq99} holds. Moreover, suppose that the process $\bar D:= D^1-D^2$  is non decreasing.
We then have $Y^1_t \geq Y^2_t$ for all $t \in [0,T]$.\\
(ii) {\em (Strict comparison theorem)}. Suppose moreover that $\gamma_t > -1$.
 If 
$Y_{t_0}=0$, for some $t_0 \in [0,T]$, then the inequalities in \eqref{eq99} are equalities and $D^1-D^2$ is constant on the time interval $[t_0,T]$.
\end{proposition}
\dproof
Using the same arguments and notation as  above, we obtain:
\begin{equation*}
  \bar{Y}_t \,\geq \,{\mathbb E}\, [\, \Gamma_{t,T} \,(\xi_1 - \xi_2) + \int_t^T 
  \Gamma_{t,s}\, (\varphi_s ds  + d\bar D_s) \mid {\cal G}_t ],\;\; 0 \leq t \leq T,\quad {\rm a.s.}\,
\end{equation*}
Hence, $\bar{Y}_t  \geq 0$ a.s.

(ii) Suppose moreover that $Y_{t_0}=0$ a.s. and that  $\gamma_t > -1$. Since $\gamma_t > -1$, we have $\Gamma_{t,T}>0$. We thus get $\xi_1=\xi_2$ a.s. and $\varphi_t=0$, $t \in [t_0,T]$ $dt \otimes dP$-a.s. Set $\Tilde{D}_t:= \int_{t_0,t} \Gamma_{t_0,s}d \bar{D}_s,$  for each $t \in [t_0,T]$. By assumption, $\Tilde{D}_T \geq 0$ a.s. and $\mathbb{E}[\Tilde{D}_T \mid {\cal G}_{t_0}]=0$ a.s. Hence $\Tilde{D}_T = 0$ a.s. Now, since 
$\Gamma_{t_0,s}>0$, $s\geq t_0$ a.s.\,, we can write $\bar{D}_T-\bar D_{t_0}= \int_{t_0,T} \Gamma^{-1}_{t_0,s}d \Tilde{D}_s$. We thus get $\bar D_T=\bar D_{t_0}$ a.s. 
\fproof

\noindent

 \section{Nonlinear pricing in a financial market with default}\label{sec3}
 \subsection{Financial market with defaultable risky asset}
  We consider a financial market which consists of one risk-free asset, whose price process $S^0$ satisfies $dS_t^{0}=S_t^{0} r_tdt$, and two risky assets with price processes $S^{1},S^{2}$ evolving according to the equations:
 \begin{eqnarray*}
dS_t^{1}&=&S_t^{1}[\mu_t^1dt +  \sigma^1_tdW_t]\\
dS_t^{2}&=&S_{t^-}^{2} [\mu^2_tdt+\sigma^2_tdW_t-dM_t],
\end{eqnarray*}
where the process $(M_t)$ is given by \eqref{M}.
Note that the second risky asset is defaultable with total default. We have $S_t^2=0$, $t \geq \vartheta$ a.s.
 
All the processes $\sigma^1,\sigma^2,$ $r, \mu^1,\mu^2$ are 
predictable (that is  ${\cal P}$-measurable). We set  $\sigma=(\sigma^1,\sigma^2)'$. 
We suppose that
$\sigma^1, \sigma^2 > 0$,
and  $r$, $\sigma^1,\sigma^2,$ 
${(\sigma^1)}^{-1}$, 
${(\sigma^2)}^{-1}$ are bounded. 

Let us  consider an investor who can invest in the three tradable assets. At time $0$, he invests the amount $x \geq 0$ in the three assets.

For $i=1,2$, we denote by $\varphi_t^i$ the amount invested in the $i^{\textit{th}}$ risky asset. Since after time $\vartheta$, the investor cannot invest his wealth in the defaultable 
asset (since its price is equal to $0$), we have $\varphi_t^2=0$ for each $t \geq \vartheta$.
%
A process $\varphi= (\varphi^1, \varphi^2)'$ belonging to ${\mathbb H}^2 \times  {\mathbb H}^2_{\lambda}$ is called a {\em risky assets stategy}.

Let $C_t$ be the cumulated cash amount which has been withdrawn from the market portfolio between time $0$ and time $t$. The process $C$ belongs to ${\cal A}^2$, that is, $C$ is an RCLL adapted non decreasing process satisfying $C_0=0$ and $E[C_T^2] < + \infty$.

The value of the associated portfolio (or {\em wealth}) at time $t$ is denoted  by $V^{x, \varphi, C}_t$. 
The amount invested in the non risky asset at time $t$ is then given by $V^{x, \varphi, C}_t - (\varphi_t^1+ \varphi_t^2)$.
\subsection{Pricing of European options with dividends in a perfect market model}
In this section, we suppose that the market model is perfect. In this case, by the self financing condition, the wealth process $V^{x, \varphi, C}$ (simply denoted by $V$) follows the dynamics:
\begin{align*}\label{portfolio}
dV_t & = r_t V_t+\varphi_t^1 (\mu^1_t - r_t)+\varphi_t^2(\mu^2_t - r_t)) dt - dC_t+
(\varphi_t^1 \sigma^1_t + \varphi_t^2 \sigma^2_t) dW_t - \varphi_t^2  dM_t\\
 & = \left(r_t V_t+(\varphi_t^1\sigma_t^1+\varphi_t^2\sigma_t^2) \theta_t^1- \varphi_t^2 \theta_t^2 \lambda_t  \right) dt - dC_t+ 
\varphi_t ' \sigma_t dW_t - \varphi_t^2  dM_t,
\end{align*}
where $\theta_t^1:=\dfrac{\mu_t^1-r_t}{\sigma_t^1}$,
$\theta_t^2:= - \dfrac{\mu_t^2-\sigma_t^2 \theta_t^1-r_t}{\lambda_t  }\,{\bf 1}_{\{t \leq \vartheta \} }$.\\
Loosely speaking, $dC_t$ represents the amount withdrawn from the portfolio during the time period 
$[t, t + dt]$.

Suppose that the processes $\theta^1$ and $\theta^2 \sqrt{\lambda}$ are bounded.



Let $T>0$. Let $\xi$ be a $\mathcal{G}_T$-measurable random variable belonging to  ${L}^2$, and let 
$D$ be a non decreasing process belonging to ${\cal A}^2$.
We consider a European option with maturity $T$, payoff $\xi$ and cumulative dividend process $D$. For each $t\in [0,T]$, 
$dD_t$ represents the dividend amount paid to the owner of the option between time $t$ and time $t+dt$.

The aim is to price this contingent claim. Let us consider a seller who wants to sell the option at time $0$. With the amount he receives at time $0$ from the buyer, he wants to be able to construct a portfolio which allows him to pay to the buyer the amount $\xi$ at time $T$ and the intermediate dividends.

By Proposition \ref{existencegeneral}, there exists an unique process $(X, Z, K) \in \mathcal{S}^2 \times {\mathbb H}^2 \times  {\mathbb H}^2_{\lambda}$ solution of the following $\lambda$-linear BSDE:
\begin{equation}\label{portfolio}
- dX_t = \displaystyle -  (r_t X_t+Z_t\theta_t^1+K_t  \theta_t^2 \lambda_t) dt + dD_t-  Z_t dW_t - K_t  dM_t\,; \quad
X_T=\xi.
\end{equation}
Note that the driver of this BSDE is given  for each $(\omega, t,y,z,k)$ by
\begin{equation}\label{vert}
g(\omega, t,y,z,k) = - r_t(\omega) y - z \theta^1_t (\omega)-  \theta^2_t (\omega)\lambda_t (\omega)\, k.
\end{equation}
Since by assumption, the coefficients $r, \sigma^2, \theta^1$, $\theta^2 \sqrt{\lambda}$  are predictable and bounded, it follows that $g$ is a $\lambda$-{\em linear} driver (see Definition \ref{deflinear}).
The solution $(X, Z, K)$ corresponds to the replicating portfolio. More precisely, 
the hedging risky assets stategy  $\varphi$ is such that 
\begin{equation} \label{st}
 {\varphi_t}' \sigma_t = Z_t \;\; ; \;\;
 - \varphi_t^2 = K_t,
\end{equation}
where ${\varphi_t}' \sigma_t= {\varphi ^1_t} \sigma^1_t + {\varphi^2_t} \sigma^2_t$.
Note that this defines a change of variables $\Phi$ defined by:\\
$\Phi:{\mathbb H}^2 \times  {\mathbb H}^2_{\lambda} \rightarrow {\mathbb H}^2 \times  {\mathbb H}^2_{\lambda}; 
(Z, K) \mapsto \Phi (Z, K):= \varphi$, where $\varphi= (\varphi^1, \varphi^2)$ is given by 
\eqref{st}, 
which is equivalent to 
\begin{equation} \label{stbis}
 \varphi_t^{2} = - {K_t} \;\; ; \;\; 
\varphi_t^{1} = \frac{Z_t -  \varphi_t^{2}  \sigma^2_t }{\sigma^1_t}= 
 \frac{Z_t +   \sigma^2_t  K_t}{\sigma^1_t}.
\end{equation}
The process $D$ corresponds to the cumulated cash withdrawal.
The process $X$ coincides with $V^{X_0, \varphi, D}$, the value of the  portfolio 
associated with initial wealth $x=X_0$, portfolio strategy $\varphi$ and cumulated (dividend) cash withdrawal $D$. 
From the seller's point of view, this portfolio is a hedging portfolio since, by investing the initial amount $X_0$ in the reference assets along the strategy $\varphi$, it allows him to pay the amount $\xi$ to the buyer at time $T$ and the intermediate dividends. 
We derive that $X_0$ is the initial price of the option, called {\em hedging price}, and denoted by 
$X_0^D(\xi)$. Similarly, for each time $t \in [0,T]$, $X_t$ is the {\em hedging price} at time $t$  of the option, and is denoted by 
$X_t^D(\xi)$. 

Since the driver $g$ given by \eqref{vert}  is $\lambda$-{\em linear}, 
the representation property of the solution of a $\lambda$-{\em linear} BSDE (see Theorem \ref{linear}) yields  
\begin{equation}\label{free-arbitrage price}
X^D_t(\xi)=\mathbb{E}[e^{-\int_t ^T r_s ds} \zeta_{t,T}\xi + \int_t^T e^{-\int_t ^s r_u du}  \zeta_{t,s} dD_s\,|\,{\cal G}_t],
\end{equation}
 where 
$\zeta$ satisfies 
\begin{equation}\label{zeta}
d\zeta_{t,s}= \zeta_{t,s^-} [-\theta^1_s dW_s - \theta^2_s  dM_s]; \quad \zeta_{t,t}=1.
\end{equation}
 This defines a {\em linear} price system $X$: $(\xi,D) \mapsto X^D(\xi)$.
Suppose now that $\theta^2_t < 1$, $0 \leq t \leq \vartheta\,$ $\,dt \otimes dP$-a.s.\,
Moroever, by Proposition \ref{roro}, the process $\zeta_{0,.}$ is a square integrable positive martingale. 
By classical results, 
the probability measure with density $\zeta_{0,T}$ on ${\cal G}_T$ is the unique 
 {\em martingale probability measure}, and
$X$ corresponds to the classical
 free-arbitrage price system (see  e.g. Proposition 7.9.11 in \cite{JYC}).

\subsection{Nonlinear pricing of European options with dividends in an imperfect market  with default} 
From 
 now on,   we assume that  there are  imperfections in the market which are taken into account via 
the {\em nonlinearity} of the
dynamics of the wealth. More precisely, 
we suppose that  the {\em wealth} process  $V^{x, \varphi,C}_t$ (or simply $V_t$)
associated with an initial wealth $x$, a strategy $\varphi=(\varphi^1, \varphi^2)$ in ${\mathbb H}^2 \times  {\mathbb H}^2_{\lambda}$ and a cumulated withdrawal process $C$ satisfies  the following dynamics:
\begin{equation}\label{riche}
-dV_t= g(t,V_t, {\varphi_t}' \sigma_t , - \varphi_t^{2} ) dt - {\varphi_t}' \sigma_t dW_t+dC_t +\varphi_t^{2} dM_t; \;  V_0=x,
\end{equation}
where $g$ is a  nonlinear $\lambda$-{\em admissible}
driver (see Definition \ref{defd}). 
Equivalently, setting $Z_t= {\varphi_t}' \sigma_t$ and
  $K_t= -  \varphi_t^2 $,
 \begin{equation}\label{wea}
-dV_t= g(t,V_t, Z_t,K_t ) dt -  Z_t dW_t+dC_t - K_t dM_t ; \;  V_0=x.
\end{equation}
Note that in the special case of a perfect market, $g$ is given by \eqref{vert}.

Let us consider a European option with maturity $T$, terminal payoff  $\xi \in {L}^2({\cal G_T})$ and dividend process $D \in \mathcal{A}^2$ in this market model. Let $(X^D(T, \xi), Z^D(T, \xi), K^D(T, \xi)),$ also denoted by $(X,Z,K)$, be 
the solution of BSDE associated with terminal time $T$, ``generalized driver" $g(\cdot)dt+dD_t$ and terminal condition $\xi$, that is satisfying   

\begin{equation*}
-dX_t = g(t,X_t, Z_t,K_t ) dt + dD_t -  Z_t dW_t - K_t dM_t; \quad
X_T=\xi.
\end{equation*}

The process $X=X^D(T, \xi)$ is equal to  the wealth process associated with initial value $x= X_0$,
strategy $\varphi $ $= \Phi  (Z, K)$ (see \eqref{stbis}) and cumulated amount $D$ of  cash withdrawals  that is
 $X= V^{X_0, \varphi,D}$.
Its initial value $X_0=X^D_0(T, \xi)$  is thus a sensible price 
 (at time $0$)  of the option for the seller since this amount allows him/her to construct a trading 
strategy  $\varphi $, called \textit{hedging} strategy,  such that the value of the associated portfolio is equal to $\xi$ at time $T$.  Moreover, the cash withdrawals perfectly replicate the dividends of the option. Similarly, $X_t=X^D_t(T, \xi)$  is a sensible price for the seller
 at time $t$. 

For each maturity $S\in [0,T]$ and for each pair ``payoff-dividend" $(\xi, D) \in 
 {L}^2({\cal G_S}) \times {\cal A}^2,$ we define the 
{\em $g$-value process}   by 
${\cal E}_{t,S}^{^{g,D}} (\xi):= X^D_t(S, \xi), $ $t \in [0,S]$. Note that ${\cal E}_{t,S}^{^{g,D}} (\xi)$ can be defined 
on the whole interval $[0,T]$ by setting ${\cal E}_{t,S}^{^{g,D}} (\xi):= {\cal E}_{t,T}^{^{g^S,D^S}} (\xi)$ for $t \geq S$, where 
$g^S(t,.):= g (t,.) {\bf 1}_{t \leq S}$ and $D_t^S := D_{t \wedge S}$. 

This  leads  to a {\em nonlinear pricing} system 

$${\cal E}^{^{g,\cdot}}: (S,\xi,D) \mapsto {\cal E}^{^{g,D}}_{\cdot,S}(\xi).$$ 

When there are no dividends, it reduces to the {\em nonlinear pricing} system ${\cal E}^{^{g,0}}$ (usually denoted by ${\cal E}^{^{g}}$), first introduced by El Karoui-Quenez (\cite{EQ96})
 in a Brownian framework 

We now give some  properties on this {\em nonlinear pricing} system $\mathcal{E}^{^{g,\cdot}}$ which generalize those given in \cite{EQ96} to the case with a default jump and dividends.

\begin{itemize}
\item[$\bullet$] \textbf{Consistency.} By the flow property for BSDEs, $\mathcal{E}^{^{g,\cdot}}$ is \textit{consistent}. More precisely, let $S'\in [0,T]$, $\xi \in L^2(\mathcal{G}_T)$, $D\in \mathcal{A}^2,$ and let $S$ be a stopping time smaller than $S'$.  Then for each time $t$ smaller than $S$, the $g$-value of the option  associated with payoff $\xi$, (cumulated) dividend process $D$ and maturity $S'$ coincides with the $g$-value of the option associated with maturity $S$, payoff $\mathcal{E}_{S,T}^{^{g,D}}(\xi)$ and dividend process $D$, that is
$$\mathcal{E}_{t,S'}^{^{g,D}}(\xi)=\mathcal{E}_{t,S}^{^{g,D}}(\mathcal{E}_{S,S'}^{^{g,D}}(\xi)) \text{ a.s. }$$
%
%
\paragraph{$\bullet$  Zero-one law.}

If $g(t,0,0,0) = 0$ \footnote{Note that when the market is perfect, $g$ is given by \eqref{vert} and thus satisfies $g(t,0,0,0) = 0$.},  then the price of the European option with null payoff and no dividends is equal to $0$. More precisely, $\mathcal{E}^{^{g,\cdot}}$ satisfies  the {\em  Zero-one law} property: 
for all maturity $S \in [0,T]$, for all payoff $\xi \in L^2(\mathcal{G}_S),$ and cumulated dividend process $D \in \mathcal{A}^2$,\\
 $\mathcal{E}_{t,S}^{^{g,D^A}}({\bf 1}_A \xi)=\
   {\bf 1}_A \mathcal{E}_{t,S}^{^{g,D }}( \xi)$ $a.s$ for $t \leq S$, $A \in {\cal G}_t$, and $\xi$ $\in$ $L^2({\cal G}_S)$, where $D^A$ is the process defined  by $D^A_s:= (D_s -D_t){\bf 1}_A
   {\bf 1}_{s\geq t}$.

Because of the presence of the default jump, the {\em nonlinear pricing} system $\mathcal{E}^{^{g,\cdot}}$ is not necessarily monotone with respect to $(\xi,D)$.
We introduce the following Assumption. 

\begin{assumption}\label{Royer} 
Assume that there exists a map \begin{equation*}
 \gamma:  [0,T]  \times \Omega\times \R^4   \rightarrow  \R \,; \, (\omega, t, y,z, k_1, k_2) \mapsto 
\gamma_t^{y,z,k_1,k_2}(\omega)
\end{equation*}
 ${\cal P } \otimes {\cal B}(\R^4) $-measurable, satisfying $ dP\otimes dt $-a.s.\,, for each $(y,z, k_1, k_2)$ $\in$ $\R^4$,
  \begin{equation*}
(\gamma_t^{y,z,k_1,k_2}
 \sqrt{ \lambda_t}) \;\text { bounded} \;\;   \text{ and } \quad\gamma_t^{y,z,k_1,k_2} \geq -1,
\end{equation*}
 and
\begin{equation} \label{critere}
g( t,y,z, k_1)- g(t,y,z, k_2) \geq  \gamma_t^{y,z, k_1,k_2} (k_1 - k_2 )  \lambda_t,
\end{equation} 

\end{assumption}
Recall that $\lambda$ vanishes after $\vartheta$ and $g(t,\cdot)$ does not depend on $k$ 
on $\{t >\vartheta\}$. Hence, the inequality \eqref{critere} is always satisfied on $\{t >\vartheta\}$.\\
Note that the above assumption holds e.g. if $g(t,\cdot)$ is non decreasing with respect to $k$, 
or if 
 $g$ is ${\cal C}^1$ in $k$  with $ \partial_k g(t, \cdot) \geq - \lambda_t$ on $\{t \leq \vartheta\}$.\\
 In the case of a perfect
  market, it is satisfied  when
  $\theta^2_t\leq 1$.


Before giving some additional properties (which hold under this Assumption), we introduce the following partial order relation, defined for each fixed time $S \in [0,T]$, on the set of pairs "payoff-dividends" by: for each $(\xi^1,D^1), (\xi^2,D^2)  \in L^2(\mathcal{G}_S) \times \mathcal{A}^2$ by
\begin{equation*}
(\xi^1,D^1)  \succ (\xi^2,D^2) \,\,\quad  \text{ if } \quad  \xi^1 \geq \xi^2\,\, {\rm a.s.} \text{ and } D^1-D^2 \text{   is non decreasing.} 
\end{equation*}
Loosely speaking, the non decreasing property of $D^1-D^2$ corresponds to the fact that the instantaneous dividends paid between times $s$ and $s+ds$ corresponding to $D^1$ are greater or equal to the ones  corresponding to $D^2$, that is $dD^1_s \geq dD^2_s$.

\item[$\bullet$] \textbf{Monotonicity.} Under Assumption \ref{Royer}, the nonlinear pricing system ${\cal E}^{^{g,\cdot}}$ is non decreasing with respect to the payoff and the dividend. More precisely, for all maturity $S \in [0,T]$, for all payoffs $\xi_1, \xi_2 \in L^2(\mathcal{G}_S),$ and cumulated dividend processes $D^1, D^2 \in \mathcal{A}^2,$ the following property holds:

 If $(\xi^1,D^1)  \succ (\xi^2,D^2)$, then we have ${\cal E}_{t,S}^{g,D^1}(\xi_1) \geq {\cal E}_{t,S}^{g,D^2}(\xi_2)$, \, $t \in [0,S]$ a.s.

This property follows from the comparison theorem for BSDEs with ``generalized drivers" (Proposition \ref{comparisongeneral} (i)) applied
  to $g^1=g^2=g$ and  $\xi^1$, $\xi^2$, $D^1$, $D^2$ (Indeed, in this case, by Assumption \ref{Royer}, Assumption \eqref{autre} holds with $\gamma_t:= \gamma_t^{Y^2_{t^-},Z_t^2,K^1_t,K^2_t}$).
  
 Using this comparison theorem, we also derive the following property:
 
\item[$\bullet$] \textbf{Convexity.} Under Assumption \ref{Royer}, if $g$ is convex with respect to $(y,z,k)$, then the nonlinear pricing system ${\cal E}^{^{g,D}}$ is convex, that is, for any $\alpha \in [0,1]$, $S \in [0,T]$, $\xi_1, \xi_2 \in L^2(\mathcal{G}_S), D^1, D^2 \in \mathcal{A}^2$
$$
\mathcal{E}_{t,S}^{^{g,\alpha D^1+(1-\alpha) D^2}}(\alpha \xi_1+(1-\alpha) \xi_2) \leq \alpha \mathcal{E}_{t,S}^{^{g,D^1}}(\xi_1)+(1-\alpha )\mathcal{E}_{t,S}^{^{g,D^2}}(\xi_2), \,\,\, \, \text{ for all } t \in [0,S].
$$

\item[$\bullet$] \textbf{Nonnegativity.} Under Assumption \ref{Royer}, when $g(t,0,0,0) \geq 0$, 
the nonlinear pricing system  ${\cal E}^{^{g,\cdot}}$ is nonnegative, that is, for each $S\in [0,T]$, for all non negative $\xi \in {L}^2({\cal G_S})$ and all $D \in \mathcal{A}^2$, 
we have ${\cal E}^{^{g,D}}_{\cdot, S} (\xi)\geq 0$ a.s.\, 
\\


Moreover, under the additional assumption $\gamma_t^{y,z,k_1,k_2}>-1$ in Assumption \ref{Royer}, 
using the strict comparison theorem (Proposition \ref{comparisongeneral} (ii)), we derive the following no arbitrage property: 
\item[$\bullet$] \textbf{No arbitrage.} Under Assumption \ref{Royer} with  $\gamma_t^{y,z,k_1,k_2}>-1$, the nonlinear pricing system ${\cal E}^{^{g,\cdot}}$ satisfies the \textit{no arbitrage} property:\\
for all maturity $S \in [0,T]$, and for all payoffs $\xi^1, \xi^2 \in L^2(\mathcal{G}_S)$, and cumulated dividend processes $D^1, D^2 \in \mathcal{A}^2$, $t_0 \in [0,S],$ and $A \in \mathcal{G}_{t_0}$,

Suppose that $(\xi^1,D^1)  \succ (\xi^2,D^2)$, and ${\cal E}_{t_0, S}^{^{g,D^1}}(\xi_1)={\cal E}_{t_0, S}^{^{g,D^2}}(\xi_2)$ a.s. on $A \in \mathcal{G}_{t_0}$.\\
Then, $\xi_1=\xi_2$ a.s. on $A$ and $(D_t^1-D_t^2)_{t_0 \leq t \leq S}$ is  a.s. constant on $A$, that is $D^1_S - D^1_{t_0} = D^2_S - D^2_{t_0}$ a.s. on $A$. In other words, the  payoffs and the instantaneous dividends paid between $t_0$ and $S$ are equal a.s. on $A$.


%

\end{itemize} 
The {\bf No arbitrage} property also ensures that when  $\gamma_t^{y,z,k_1,k_2}>-1$, the nonlinear pricing system $\mathcal{E}^g$ is strictly monotone. 
Note that when the market is perfect, the condition $\gamma_t^{y,z,k_1,k_2}>-1$ is satisfied when $\theta^2_t <1$.

 \begin{remark}
Several authors have studied dynamic risk measures defined as 
the solutions of BSDEs 
(see e.g. \cite{Peng2004, BEK, QS1}). In our framework with a default jump, given a $\lambda$-{\em admissible} driver, one can  define a {\em dynamic measure of risk} $\rho^{g}$ as follows: for each 
$S \in [0,T]$ and $\xi \in L^2({\cal G}_S)$, we set
$$\rho^{g}_\cdot (\xi, S) =  -\mathcal{E}^{g}_{\cdot,S} (\xi),$$
where $\mathcal{E}^{g}_{\cdot,S} (\xi)$ denotes the solution of the BSDE associated with terminal condition $\xi$, terminal time $T$ and driver $g$.
Then, by the results of this section, the dynamic risk-measure $\rho^{g}$ satisfies analogous properties to the ones of the nonlinear pricing system $\mathcal{E}^{g}_{\cdot,S} = \mathcal{E}^{g,0}_{\cdot,S}$ (corresponding to the case with no dividends).

\end{remark}

We now introduce the definition of an $\mathcal{E}^{^{g,D}}$-supermartingale  which generalizes the classical notion of $\mathcal{E}^g$-supermartingale.

  \begin{definition}\label{defmart}
Let  $D \in \cal A^2$ and $Y \in \cal S^2$. The process $Y$ is said to be a $\mathcal{E}^{^{g,D}}$-supermartingale (resp. $\mathcal{E}^{^{g,D}}$-martingale)  if ${\cal E}_{\sigma ,\tau}^{^{g,D}}(Y_{\tau}) \leq Y_{\sigma}$ (resp. $= Y_{\sigma}$) a.s. on $\sigma \leq \tau$,  for all $ \sigma, \tau \in \mathcal{T}_0$. 
\end{definition}


\begin{proposition}  \label{rima}
For all $S\in [0,T]$, payoff $\xi \in {L}^2({\cal G_S})$ and dividend process $D \in {\cal A}^2$, the associated $g$-value process 
${\cal E}_{\cdot,S}^{^{g,D}} (\xi)$ is an $\mathcal{E}^{^{g,D}}$-martingale.

Moreover, for all $x \in \mathbb{R}$, portfolio strategy $\varphi$ $\in$ ${\mathbb H}^2\times 
{\mathbb H}^2_{\lambda}$ and cash withdrawal process $D \in {\cal A}^2$, the associated wealth process $V^{x, \varphi,D}$ is an $\mathcal{E}^{^{g,D}}$-martingale.
\end{proposition}

\dproof The first assertion follows from the consistency property of $\mathcal{E}^{^{g,D}}$. The second one is obtained by noting that 
$V^{x, \varphi,D}$ is the solution of the 
BSDE with ``generalized driver" $g(\cdot)dt+dD_t$, terminal time $T$ and terminal condition $V_T^{x, \varphi,D}$.
\fproof

\paragraph{Example (Large investor seller)} 
When the seller is a large trader, his hedging portfolio may  affect the prices of the risky assets and the default probability. He may take into account these feedback effect in his market  model as follows. 

In order to simplify the presentation, we consider the case when the seller's strategy affects only the default intensity. 
We are given a family of probability measures parametrized by $V$ and $\varphi$.
More precisely, for each $V \in \mathcal{S}^2$ and  $\varphi \in  {\mathbb H}^2$, let $Q^{V, \varphi}$ be the probability measure equivalent to $P$, which admits $L^{V, \varphi }$ as density with respect to $P$, where $(L^{V, \varphi })$ is the solution of the following SDE:
$$
dL_t^{V, \varphi }=L_{t^-} \gamma(t,{V}_{t^-}, \varphi_t)dM_t; \quad L^{V,\varphi}_0=1.
$$
Here,
$\gamma: (\omega,t ,y, \varphi_1, \varphi_2) \mapsto \gamma  (\omega, t ,y, \varphi_1, \varphi_2) $ is  a  ${\cal P}\otimes {\cal B}({\bf R}^3)/\mathcal{B}({\bf R})$-measurable function defined on $ \Omega \times [0,T] \times  {\bf R}^2$ with $\gamma (t, \cdot)   >-1$, and such that
$(\gamma (t, \cdot)\sqrt{\lambda_t})$ is uniformly bounded. Note that by Proposition \ref{roro}, we have 
$L \in {\cal S}^2$.

By Girsanov's theorem, the process $W$ is a  $Q^{V, \varphi}$-Brownian motion and the process $M^{V, \varphi}$ defined as 
\begin{equation}\label{Girsanov}
M^{V, \varphi}_t:= N_t- \int_0^t \lambda_s (1+ \gamma(s,V_s, {\varphi}_s))ds= M_t - \int_0^t \lambda_s \gamma(s,V_s, {\varphi}_s)ds
\end{equation}
is a  $Q^{V, \varphi}$-martingale. 
 Hence, under $Q^{V, \varphi}$, the ${\mathbb G}$-default intensity process is equal to 
$\lambda_t (1+ \gamma(t,V_t, {\varphi}_t))$.
The process $\gamma(t,V_t, {\varphi}_t)$ represents the {\em impact of the seller's strategy on the default intensity}.

The dynamics of the wealth process associated with an initial wealth $x$ and  a risky assets stategy $\varphi$ 
satisfy 
\begin{equation}\label{vvv}
dV_t  = \left(r_t V_t+(\varphi_t^1\sigma_t^1+\varphi_t^2\sigma_t^2) \theta_t^1- \varphi_t^2 \theta_t^2 \lambda_t  \right)  dt - dC_t+ 
\varphi_t ' \sigma_t dW_t - \varphi_t^2  dM^{V, \varphi}_t,
\end{equation}
Let us show that this model can be seen as a particular case of the model described above associated with an appropriate map $\lambda$-admissible driver $g$. 
First, note that the dynamics of the wealth \eqref{vvv} can be written
\begin{equation*}
dV_t  = \left(r_t V_t+(\varphi_t^1\sigma_t^1+\varphi_t^2\sigma_t^2) \theta_t^1- \varphi_t^2 \theta_t^2 \lambda_t  +  \gamma(t,V_t, {\varphi}_t)  \lambda_t   \varphi_t^2 \right)dt - dC_t+ \varphi_t ' \sigma_t dW_t - \varphi_t^2  dM_t,
\end{equation*}
Equivalently, setting $Z_t= {\varphi_t}' \sigma_t$ and
  $K_t= -  \varphi_t^2 $,
 \begin{equation}\label{wea}
-dV_t= g(t,V_t, Z_t,K_t ) dt -  Z_t dW_t+dC_t - K_t dM_t,
\end{equation}
where 
$$g(t,y,z,k) =  -r_t y - z \theta^1_t -  \theta^2_t \lambda_t k + \gamma \left(t,y, (\sigma^1_t)^{-1}(z + { \sigma^2_t} k), -k \right)\lambda_t  k.
 $$
 We are thus led to the general model described above associated with this driver.

This model can be easily generalized to the case when the coefficients $\mu^1$, $\sigma^1$, $\mu^2$, $\sigma^2$ also depend on the hedging cost $V$ (equal to the price of the option) and  on the hedging strategy $\varphi^2$.
\footnote{The coefficients may also depend  on $\varphi= (\varphi^1, \varphi^2)$, but in this case, we have to assume that   the map  $\Psi:$ $(\omega, t,y,\varphi) \mapsto (z,k)$ with $z={\varphi}' \sigma_t(\omega,t,y,\varphi)$ and
 $k=- \varphi^2$ is one to one with respect to $\varphi$, and such that its inverse $\Psi^{-1}_{\varphi} $ is ${\cal P}\otimes {\cal B} ({\bf R}^3)$-measurable.  }

\appendix
\section{Appendix}

For $p \geq 2$, we introduce the spaces ${\cal S}^{p}$,  ${\mathbb H}^p$ and ${\mathbb H}^p_{\lambda}$ defined as follows.\\
Let ${\cal S}^{p}$ 
be the set of ${\mathbb G}$-adapted RCLL processes $\varphi$ such that $\mathbb{E}[\sup_{0\leq t \leq T} |\varphi_t | ^p] < +\infty$.\\
Let ${\mathbb H}^p$  be the set of ${\mathbb  G}$-predictable processes such that
 $
 \| Z\|_p^p:= \mathbb{E}\Big[(\int_0^T|Z_t|^2dt)^{p/2}\Big]<\infty \,.$\\
 Let ${\mathbb H}^p_{\lambda}$ be the set of ${\mathbb G}$-predictable processes such that
$\| U\|_{p,\lambda}^p:=\mathbb{E}\Big[(\int_0^T|U_t|^2\lambda_tdt)^{p/2}\Big]<\infty \,.$

%

\paragraph{BSDEs with a default jump in $L^p$}

\begin{proposition}
\label{p}
Let $p\geq 2$ and let $T >0$. Let  $g$ be a $\lambda$-admissible driver such that $g(t,0,0,0)$ $\in$ $\H^{p}$. Let $\xi \in {L}^p({\cal G_T})$. There exists a unique solution $(Y, Z, K)$  in 
$ \mathcal{S}^p \times {\mathbb H}^p \times  {\mathbb H}^p_{\lambda}$ of  the BSDE with default
\eqref{BSDE}. 
\end{proposition}

\begin{remark}\label{faible}
The above result still holds in the case  when there is a ${\mathbb G}$-martingale representation theorem 
with respect to $W$ and $M$,
 even if ${\mathbb G}$ is not generated by $W$ and $M$.
 \end{remark}

\dproof
The proof relies on the same arguments as in the proof of Proposition A.2 in \cite{QS1} together with the arguments used 
in the proof of Proposition \ref{existence}.
\fproof

\paragraph{BSDEs with a default jump and change of probability measure}

\quad

Let $(\beta_s)$ and $(\gamma_s)$ be  two real-valued ${\mathbb G}$-predictable processes such 
that 
$\int_0^T  (\beta_r ^2 + \gamma_r ^2  \lambda _r)\,dr$ is bounded.

 Let $(\zeta_s)$ be the process satisfying the forward SDE: 
$$\quad d \zeta_{s}  = \zeta_{s^-} (\beta_s d W_s +\gamma_s dM_s),$$ with $\zeta_0 =1$. By Proposition \ref{roro}, 
$\zeta$ is a $p$-integrable martingale, that is $\zeta_T \in L^p$ for all $p\geq 1$. We suppose that $\gamma > -1$, which implies that $\zeta_s >0$, $s \in  [0,T]$ a.s.
Let $Q$  be the probability measure equivalent to 
$P$ which admits 
 $\zeta_{T}$ as density with respect to $P$ on ${\cal G}_{T}$.

By Girsanov's theorem (see \cite{JYC} Chapter 9.4 Corollary 4.5), the process $W^{\beta}_t := W_t - \int_0^t \beta_s ds$ is a $Q$-Brownian motion and the process $M^{\gamma}$ defined as 
\begin{equation}\label{Girsanov}
M^{\gamma}_t  := M_t-\int_0^t\lambda_s\gamma_s ds= N_t-\int_0^t \lambda_s(1+\gamma_s) ds
\end{equation}
is a $Q$-martingale. 
We now show a representation theorem for $(Q, {\mathbb  G})$-local martingales with respect to $W^{\beta}$ and $M^{\gamma}$.

\begin{proposition}
Let  $m= (m_t)_{0\leq t \leq T}$ be a $(Q, {\mathbb  G})$-local martingale. There exists a unique pair of predictable processes $(z_t, k_t)$ such that 
\begin{equation}\label{alpharepres}
 m_t= m_0 + 
\int_0^t  z_s d W^{\beta}_s + \int_0^t  k_s dM^{\gamma}_s \quad 0 \leq s \leq T\quad {\rm a.s.}
\end{equation}
\end{proposition}

\dproof Since $m$ is a $Q$-local martingale, the process $\bar m_t:=\zeta_t m_t$ is a $P$-local martingale. 
By the martingale representation theorem (Lemma \ref{theoreme representation}), there exists a unique pair of predictable processes $(Z,K)$  such that 
\begin{equation*}
 \bar m_t= \bar m_0 + 
\int_0^t  Z_s d W_s + \int_0^t K_s dM_s \quad 0 \leq t \leq T\quad {\rm a.s.}
\end{equation*}
Then, by applying It\^o's formula to $m_t= \bar m_t (\zeta_t)^{-1}$ and by classical computations, one can derive the existence of $(z,k)$ satisfying $(\ref{alpharepres})$.
\fproof

From this result together with Proposition \ref{p} and Remark \ref{faible}, we derive the following corollary.
\begin{corollary}
Let $p\geq 2$ and let $T >0$. Let  $g$ be a $\lambda$-admissible driver such that $g(t,0,0,0)$ $\in$ $\H^{p}_Q$. Let $\xi \in {L}_Q^p({\cal G_T})$. There exists a unique solution $(Y, Z, K)$  in 
$ \mathcal{S}_Q^p \times {\mathbb H}_Q^p \times  {\mathbb H}^p_{Q, \lambda}$ of  the BSDE with default:
\begin{equation*}
-dY_t = g(t,Y_t, Z_t,K_t ) dt  -  Z_t W^{\beta}_t - K_t dM^{\gamma}_t; \quad
Y_T=\xi.
\end{equation*}
\end{corollary}

\end{document}